\documentclass{article}
\usepackage[utf8]{inputenc}
\usepackage[margin=1in]{geometry}
\usepackage{complexity}
\usepackage{mathtools,amsmath,amssymb}
\usepackage{physics}
\usepackage{todonotes}
\usepackage{multirow}
\usepackage{subcaption}
\usepackage{blkarray}
\usepackage{tikzit}

\tikzstyle{gate}=[fill={rgb,255: red,222; green,222; blue,222}, draw=black, shape=rectangle, minimum width=0.75cm, minimum height=0.75cm]
\tikzstyle{medgate}=[fill={rgb,255: red,222; green,222; blue,222}, draw=black, shape=rectangle, minimum width=0.5cm, minimum height=0.6cm]
\tikzstyle{minigate}=[fill={rgb,255: red,222; green,222; blue,222}, draw=black, shape=rectangle, minimum width=0.5cm, minimum height=0.5cm]
\tikzstyle{state}=[fill={rgb,255: red,222; green,222; blue,222}, regular polygon, regular polygon sides=3, draw, text width=1em, inner sep=0.5mm, outer sep=0mm, shape border rotate=90]
\tikzstyle{ministate}=[fill={rgb,255: red,222; green,222; blue,222}, regular polygon, regular polygon sides=3, draw, font={\footnotesize}, text width=0.7em, inner sep=0.2mm, outer sep=0mm, shape border rotate=90]
\tikzstyle{miniket}=[fill={rgb,255: red,222; green,222; blue,222}, regular polygon, regular polygon sides=3, draw, font={\footnotesize}, text width=0.7em, inner sep=0.2mm, outer sep=0mm, shape border rotate=270]
\tikzstyle{tinystate}=[fill={rgb,255: red,222; green,222; blue,222}, regular polygon, regular polygon sides=3, draw, font={\footnotesize}, text width=0.6em, inner sep=0.2mm, outer sep=0mm, shape border rotate=90]
\tikzstyle{tinyket}=[fill={rgb,255: red,222; green,222; blue,222}, regular polygon, regular polygon sides=3, draw, font={\footnotesize}, text width=0.6em, inner sep=0.2mm, outer sep=0mm, shape border rotate=270]
\tikzstyle{permred}=[fill={rgb,255: red,248; green,206; blue,204}, draw=black, shape=circle, font={\scriptsize}, text width=0.7em, inner sep=0.3mm, outer sep=0mm]
\tikzstyle{permblue}=[fill={rgb,255: red,218; green,232; blue,252}, draw=black, shape=circle, font={\scriptsize}, text width=0.7em, inner sep=0.3mm, outer sep=0mm]
\tikzstyle{permgreen}=[fill={rgb,255: red,218; green,252; blue,232}, draw=black, shape=circle, font={\scriptsize}, text width=0.7em, inner sep=0.3mm, outer sep=0mm]
\tikzstyle{plauette}=[fill={rgb,255: red,169; green,196; blue,235}, regular polygon, regular polygon sides=3, draw, font={\footnotesize}, text width=1.3em, inner sep=0.3mm, outer sep=0mm, shape border rotate=90]
\tikzstyle{plauetterotate}=[fill={rgb,255: red,169; green,196; blue,235}, regular polygon, regular polygon sides=3, draw, font={\footnotesize}, text width=1.3em, inner sep=0.3mm, outer sep=0mm, shape border rotate=270]
\tikzstyle{miniplauette}=[fill={rgb,255: red,169; green,196; blue,235}, regular polygon, regular polygon sides=3, draw, font={\footnotesize}, text width=0.5em, inner sep=0.2mm, outer sep=0mm, shape border rotate=90]

\usetikzlibrary{shapes}

\usepackage{graphicx}

\usepackage{xcolor}
\RequirePackage[colorlinks=true]{hyperref}
\hypersetup{
	linkcolor=[rgb]{0,0,0.5},
	citecolor=[rgb]{0, 0.5, 0},
	urlcolor=[rgb]{0.5, 0, 0}
}
\usepackage{dsfont}

\usepackage{algorithmicx}
\usepackage{algorithm}
\usepackage[noend]{algpseudocode}

\algnewcommand\algorithmicinput{\textbf{Input:}}
\algnewcommand\Input{\item[\algorithmicinput]}

\algnewcommand\algorithmicoutput{\textbf{Output:}}
\algnewcommand\Output{\item[\algorithmicoutput]}

\usepackage{amsthm}
\usepackage{thmtools,thm-restate}

\numberwithin{equation}{section}
\declaretheoremstyle[bodyfont=\it,qed=\qedsymbol]{noproofstyle}

\declaretheorem[name=Observation,numbered=no]{observation*}

\declaretheorem[numberlike=equation]{theorem}

\declaretheorem[name=Theorem,numbered=no]{theorem*}

\declaretheorem[numberlike=equation]{lemma}
\declaretheorem[name=Lemma,numbered=no]{lemma*}

\declaretheorem[numberlike=equation]{corollary}
\declaretheorem[name=Corollary,numbered=no]{corollary*}

\declaretheorem[name=Proposition,numbered=no]{proposition*}

\declaretheorem[name=Claim,numbered=no]{claim*}

\declaretheorem[name=Conjecture,numbered=no]{conjecture*}

\declaretheorem[name=Question,numbered=no]{question*}

\declaretheoremstyle[bodyfont=\it,qed=$\lozenge$]{defstyle} 

\declaretheorem[numberlike=equation,style=defstyle]{definition}
\declaretheorem[unnumbered,name=Definition,style=defstyle]{definition*}

\declaretheorem[unnumbered,name=Example,style=defstyle]{example*}

\declaretheorem[unnumbered,name=Notation=defstyle]{notation*}

\declaretheorem[unnumbered,name=Construction,style=defstyle]{construction*}

\declaretheoremstyle[]{rmkstyle}

\newcommand{\N}{\mathbb{N}}
\newcommand{\Complex}{\mathbb{C}}

\newcommand{\Exp}{\mathop{\mathbb{{E}}}}
\newcommand{\Var}{\mathop{\mathsf{Var}}}
\newcommand{\dw}{\textsf{dw}}

\DeclarePairedDelimiter\floor{\lfloor}{\rfloor}

\newcommand{\cD}{\mathcal{D}}
\newcommand{\cE}{\mathcal{E}}
\newcommand{\cF}{\mathcal{F}}
\newcommand{\cS}{\mathcal{S}}

\title{Spoofing Linear Cross-Entropy Benchmarking in Shallow Quantum Circuits}

\author{Boaz Barak\thanks{Harvard University, \texttt{b@boazbarak.org}. Supported by NSF awards CCF 1565264 and CNS 1618026 and a Simons Investigator Fellowship.} \and
Chi-Ning Chou\thanks{Harvard University, \texttt{chiningchou@g.harvard.edu}. } 
\and
Xun Gao\thanks{Harvard University, \texttt{xungao@g.harvard.edu}.}
}

\date{\today}

\begin{document}

\maketitle
\thispagestyle{empty}

\begin{abstract}
The \emph{linear cross-entropy benchmark} (Linear XEB) has been used as a test for procedures simulating quantum circuits. Given a quantum circuit $C$ with $n$ inputs and outputs and purported simulator whose output is distributed according to a distribution $p$ over $\{0,1\}^n$, the linear XEB fidelity of the simulator is  $\cF_{C}(p) = 2^n \mathbb{E}_{x \sim p} q_C(x) -1$ where $q_C(x)$ is the probability that $x$ is output from the distribution $C \ket{0^n}$. A trivial simulator (e.g., the uniform distribution) satisfies  $\cF_C(p)=0$, while Google's noisy quantum simulation of a 53 qubit circuit $C$  achieved a fidelity value of $(2.24\pm0.21)\times10^{-3}$ (Arute et. al., Nature'19).

In this work we give a classical randomized algorithm that for a given circuit $C$ of depth $d$ with Haar random 2-qubit gates achieves in expectation a fidelity value of $\Omega(\tfrac{n}{L} \cdot 15^{-d})$ in running time $\poly(n,2^L)$.  Here $L$ is the size of the \emph{light cone} of $C$: the maximum number of input bits that each output bit depends on. In particular, we obtain a polynomial-time algorithm that achieves large fidelity of $\omega(1)$ for depth $O(\sqrt{\log n})$ two-dimensional circuits.
To our knowledge, this is the first such result for two dimensional circuits of super-constant depth.
Our results can be considered as an evidence that fooling the linear XEB test might be easier than achieving a full simulation of the quantum circuit. 
\end{abstract}

\newpage

\thispagestyle{empty}

\setcounter{page}{0}
\tableofcontents

\newpage

\section{Introduction}

\emph{Quantum computational supremacy} refers to experimental violations of the extended Church Turing Hypothesis using quantum computers.
The most famous (and arguably at this point the only) example of such an experiment was carried out by Google~\cite{google19}. The Google team constructed a device $D$ that provides a ``noisy simulation'' of a quantum circuit $C$ with $n$ inputs and $n$ outputs.
The device can be thought as a ``black box'' that samples from a distribution $p^D$ over $\{0,1\}^n$ that (loosely) approximates the distribution $q_C$ that corresponds to measuring $C$ applied to the all-zeroes string $0^n$.
The quality of the device was measured using a certain benchmark known as the \emph{Linear Cross-Entropy} benchmark (\textit{a.k.a.} Linear XEB).
The computational hardness assumption underlying the experiment is that no efficient classical algorithm can achieve a similar score.
In this paper we investigate this assumption, giving a new classical algorithm for ``spoofing'' this benchmark in certain regimes.
While our algorithm falls short of spoofing the benchmark in the parameter regime corresponding to the Google experiment, we do manage to achieve non-trivial results for deeper circuits than were known before. 
To our knowledge, this is the first algorithm that directly targets the linear XEB benchmark, without going through a full simulation of the underlying quantum circuit.
Thus our work can be viewed as evidence that obtaining non-trivial performance for this benchmark is \emph{not} equivalent to simulating quantum circuits.

The linear XEB benchmark is defined as follows. Let $C$ be an $n$-qubit quantum circuit and $q_C:\{0,1\}^n\rightarrow[0,1]$ be the pdf of the distribution obtained by measuring $C|0^n \rangle$. For each $x\in\{0,1\}^n$, the instance linear XEB of $x$ is defined as $\cF_C(x):=2^nq_C(x)-1$.
For every probability distribution $p$, the linear XEB fidelity of $p$ with respect to circuit $C$ is defined as
\[
\cF_{C}(p) := \Exp_{x\sim p}\left[\cF_C(x)\right] = 2^n\sum_{x\in\{0,1\}^n}q_C(x)p(x)-1 \, .
\]

If $C$ is a fully random circuit, then in expectation a perfect simulation $p=q_C$ achieves $\cF_C(p)=2$.\footnote{This follows since $q_C$ is the \emph{Porter Thomas distribution}.  However, $q_C$ is \emph{not} the maximizer of $\cF_C(p)$: a distribution $p$ that has all its mass on the mode $x$ of the distribution $q_C$ will achieve $\cF_C(p) \geq \Omega(n)$ for fully random circuits, and  even higher values for shallower circuits as we'll see below.}
Google's ``quantum computational supremacy'' experiment demonstrated a noisy simulator sampling from a distribution $p$ with $\cF_C(p) \approx (2.24\pm0.21)\times10^{-3}$  for two dimensional 53-qubit circuits of depth 20.
A trivial simulation (e.g. a distribution $p$ which is the uniform distribution or another distribution independent of $C$) will achieve $\cF_C(p)=0$.
Motivated by the above, we say that $p$ achieves \emph{non trivial fidelity} with respect to the circuit $C$ if $\cF_C(p) = 1/\poly(n)$.\footnote{As mentioned above, for an ideal simulation in random circuits the Fidelity will be a constant. For noisy quantum circuits such as Google's, the fidelity is roughly $\exp(-\epsilon s)$ where $\epsilon$ is the level of noise per gate  and $s = \Theta(d \cdot n)$ is the number of gates in the circuit.}

The computational assumption underlying quantum computational supremacy with respect to some distribution $\cD$ over quantum circuits can be defined as follows. 
For every efficient randomized classical algorithm $A$, with high probability over $C \sim \cD$ , if we let $A_C$ be the distribution of $A$'s output on input $C$, then $\cF_C(A_C) = n^{-\omega(1)}$. That is, the distribution output by $A(C)$  has trivial fidelity with respect to $C$.
Aaronson and Gunn~\cite{AG19} showed that this assumption follows from a (very strong) assumption they called ``Linear Cross-Entropy Quantum Threshold Assumption'' or XQUATH.\footnote{While \cite{AG19} state their result for $\Omega(1)$  fidelity,  their proof shows that the XQUATH assumption implies that classical algorithms can not achieve $1/\poly(n)$ empirical fidelity with $poly(n)$ samples.}

In this work, we present an efficient classical algorithm $A$ that  satisfies $\cF_C(A_C)=\Omega(1)$ for  quantum circuit $C$ sampled from a distribution with Haar random 2-qubit gates with small \emph{light cones} (see Definition~\ref{def:light cone}).\footnote{If $C$ is a quantum circuit and $i$ is an output bit of $C$, then the \emph{light cone} of $i$ is the set of all input bits $j$ that are connected to $i$ via a path in the circuit. For general circuits the light cone size can be exponential in the depth, but for one or two dimensional circuits, of the type used in quantum supremacy experiment, the light cone size is polynomial in the depth.}
Specifically, we prove the following theorem:

\begin{restatable}[Linear XEB for circuits with small light cones]{theorem}{maintheorem}\label{thm:general}
Let $n,d,L\in\N$ and let $\cD$ be a distribution over $n$-qubit quantum circuits with (i) light cone size at most $L$, (ii) depth at most $d$, and (iii) Haar random $2$-qubit gates. Then, there exists a classical randomized algorithm $A$ running in $\poly\left(n,2^{L}\right)$ time such that
\[
\Exp_{C\sim\cD}\left[\cF_C(A_C)\right]\geq\left(1+15^{-d}\right)^{\floor*{\frac{n}{L}}}-1 \, .
\]
\end{restatable}

For constant dimensional circuits (such as the 2D quantum architecture used by Google), Theorem~\ref{thm:general} yields the following corollary:

\begin{corollary}[Constant dimensional circuits]\label{cor:lowdim}
Let $n\in\N$ and $d= O(\log n)$. Let $c\in\N$ be a constant and $\cD$ be the distribution of $n$-qubit $c$-dimensional circuits of depth $d$ with Haar random $2$-qubit gates. There is a randomized algorithm $A$ running in time $2^{O(d^c)}$ such that
$$\mathbb{E}_{C\sim\cD} \cF_C(A_C) = 1/\poly(n) \;.$$
\end{corollary}

\begin{proof}
A $c$-dimensional of depth $d$ circuit has light-cone of size $L=O(d^c) = n^{o(1)}$ for $d=O(\log n)$.
Let $d = \alpha \log n$.
By plugging in the parameters of Theorem~\ref{thm:general}, we see that (using $\log_2 15 < 4$ and $n/L \geq n^{1-o(1)}$) the expected value of the fidelity is at least
$$
(1 + 2^{-4 d})^{n^{1-o(1)}} - 1  \geq \Omega(\tfrac{n^{1-o(1)}}{n^{4\alpha}}) \;.
$$
The right hand size is at least $1/\poly(n)$ for every constant $\alpha$ and in fact $\omega(1)$ for $\alpha < 1/4$.
\end{proof}

The bounds of Corollary~\ref{cor:lowdim} do \emph{not} correspond to the Google experiment where the depth is roughly comparable to $\sqrt{n}$, rather than logarithmic.
However, prior works in the literature were only able to achieve good linear XEB performance for circuits of \emph{constant} depth (see Section~\ref{sec:prior}).
More importantly (in our view) is that our bounds show that it may be possible to achieve good linear XEB performance without achieving a full simulation.

\subsection{From expectation to concentration.} 

In actual experiments, one measures the \emph{empirical} linear XEB, obtained by sampling $x_1,\ldots,x_T$ independently from the distribution $p$ and computing $\tfrac{1}{T} \sum_{i=1}^T 2^n q_C(x_i) - 1$.
Thus in our classical simulation we want to go beyond achieving large \emph{expected} linear XEB benchmark, to show that our algorithm $A$ actually achieves non-trivial empirical linear XEB with probability at least inverse polynomial over the choice of the circuit and with a number of samples $T$ that is at most polynomial in $n$.
These probability bounds are more challenging to prove, and at the moment our results are weaker than the optimal bounds one can hope for.

\paragraph{Probability over circuits.}
For bounding the probability over \emph{circuits} we show in Section~\ref{sec:probcircuits}, that in the setting of Theorem~\ref{thm:general} , for logarithmic depth circuits, we can obtain $1/\poly(n)$ fidelity with probability at least $1/\poly(n)$. 
We also obtain more general tradeoffs between the fidelity, probability, and depth, see \autoref{cor:generalprob}.
We conjecture that random circuits from the distributions we consider exhibit much better concentration, and fact that the  fidelity sharply concentrates around its expectation.

\paragraph{Sample complexity, or probability over the algorithms' randomness.}
Bounding the \emph{sample complexity} of our algorithm is a more difficult task then the expectation analysis because it requires higher moment information on $\cF_C(A_C)$.
We obtain only partial bounds in this setting, which we believe to be far from optimal.
In Section~\ref{sec:sample complexity} we show that an upper bound for the \textit{collision probability} of $q_C$ is sufficient to give an upper bound for the sample complexity of our algorithm.
Specifically, letting $CP(q) = \sum_{x \in \{0,1\}^n} q(x)^2$, we show that if $CP(q) = M \cdot 2^{-n}$ then the number of samples needed for the empirical linear XEB to achieve a value of at least $\epsilon$  is $M \cdot \exp(O(\epsilon  \cdot 15^d))$.
In particular, for logarithmic depth circuits we can get inverse-polynomial empirical fidelity using  $O(M)$ samples.
For \emph{random} quantum circuits, where $q_C$ is the Porter-Thomas distribution (with $q_C(x)$ drawn independently as the square of a mean zero variance $2^{-n}$ normal variable), it is known that $CP(q_C) = O(2^{-n})$, i.e., $M=O(1)$. 
For shallow circuits, of the type we study, we show in \autoref{lem:anti concentration of 1D} that $CP(q_C) = O(2^{-n})$ for  random \emph{one dimensional} circuits of depth at least $c \log n$ for some constant $c>0$, which shows that we can achieve for such circuits inverse polynomial empirical fidelity using a polynomial number of samples. 
While this is significantly more technically challenging to prove, we conjecture that the same  collision probability bound holds for \emph{two dimensional} circuits of depth $\Omega(\sqrt{\log n})$. This conjecture, if true, will imply that for such circuits we can achieve $1/\poly(n)$ empirical fidelity using a polynomial number of samples, and constant fidelity using a sub-exponential (e.g. $\exp{\exp{O(\sqrt{\log n})}}$) number of samples.

\subsection{Prior works} \label{sec:prior}

Prior classical algorithms mostly focused on the task of obtaining a full simulation (sampling from $C|0^n\rangle$  or from a distribution close to it in statistical distance).
We are not aware of any prior work that directly targeted the linear XEB measure and gave explicit bounds for the performance in this measure that are not implied by approximating the full distribution.

Napp et al~\cite{napp2019efficient} gave an algorithm to simulate random two-dimensional circuits of some small constant depth. They gave strong theoretical evidence that up to a certain constant depth, such circuits can be approximated by 1D circuits of small entanglement (i.e., ``area law'' as opposed to ``volume law''), which can be effectively simulated using Matrix Product States \cite{vidal2004efficient}.
However, \cite{napp2019efficient} also gave evidence that the system undergoes a \emph{phase transition} when the depth is more than some constant size (around $4$), at which case the entanglement grows according to a ``volume law'' and hence their methods cannot be used to simulate circuits of super-constant depth.

Another direction of approximating large  quantum circuits has  considered the effect of \emph{noise}. Some restricted classes of noisy quantum circuits were shown to be simulated by polynomial time classical algorithms in \cite{bremner2017achieving,yung2017can} (in contrast to their noiseless variant  \cite{bremner2016average}).
This was also extended to more general random circuits by  \cite{gao2018efficient}. 
Very recent work has given numerical results suggesting that states generated by noisy quantum circuits could be approximated by Matrix Product States or Operators under the state fidelity measure  \cite{zhou2020limits,noh2020efficient}.\footnote{\cite{zhou2020limits} briefly discusses the linear XEB measure as well, see Figure 7 there.}
Low degree Fourier expansions yield other candidates for approximating such quantum states \cite{gao2018efficient,bremner2017achieving}.

\section{Preliminaries}\label{sec:prelim}

In this section we introduce some of the notions we use for quantum circuits, and in particular distributions of random quantum circuits of fixed architecture, as well as tensor networks for analyzing quantum circuits.  We include this here since some of this notation, and in particular tensor networks, might be unfamiliar to theoretical computer science audience.
However, the reader can choose to skip this section and refer back to it as needed.
We also record some useful facts of tensor networks and quantum circuits in \autoref{app:ommitted}.

For $n\in\N$, an $n$-qubit quantum state $\ket{\psi}=\sum_{x\in\{0,1\}^n}\alpha_x\ket{x}$ is a unit vector in $\Complex^{2^n}$. We let $I,X,Y,Z$ denote the Pauli matrices where
\[
I=\begin{bmatrix}1&0\\0&1\end{bmatrix} ,\ 
X=\begin{bmatrix}0&1\\1&0\end{bmatrix} ,\ 
Y=\begin{bmatrix}0&i\\-i&0\end{bmatrix} ,\ \text{and }
Z=\begin{bmatrix}1&0\\0&-1\end{bmatrix} \, . 
\]

The following definition captures the notion of an ``architecture'' of a quantum circuit (see \autoref{fig:circuit skeleton} for an example):

\begin{definition}[Circuit skeleton and light cone] \label{def:skeleton}
Let $n\in2\N$ and $d\in\N$, an $n$-qubit depth $d$ \emph{circuit skeleton} $\cS$ is a directed acyclic graph with $d+2$ layers with the following structure. For convenience, we start the index of layers from $0$.
\begin{itemize}
\item The $0^\text{th}$ and the $(d+1)^\text{th}$ layer has $n$ nodes corresponding to the $n$ input and output qubits. Each node in the first layer has exactly one out-going edge to the next layer while each node the last layer has exactly one in-going edge from the previous layer.
\item Each of the other layers has exactly $n/2$ nodes and each node has exactly two in-going to the next layer and two out-going edges from the previous layer. Specifically, the first edge $i^\text{th}$ gate is indexed by $2i-1$ while the second edge is indexed by $2i$ for each $i=1,2,\dots,n/2$.
\item For each $i=1,2,\dots,n$, the $i^\text{th}$ input node connects to the $i^\text{th}$ edge of the second layer while the $i^\text{th}$ edge of the $(d+1)^\text{th}$ layer connects to the $i^\text{th}$ output node.
\end{itemize}
Note that with the above definition, a circuit skeleton $\mathcal{S}$ can be specified by $d+1$ many permutations $\pi^{(0)},\pi^{(1)},\dots,\pi^{(d)}\in S_n$. Namely, for each $t=0,1,2,\dots,d-1$ and $i=1,2,\dots,n$, the $i^\text{th}$ edge of the $t^\text{th}$ layer connects to the $\pi^{(t)}(i)^\text{th}$ edge of the $(t+1)^\text{th}$ layer.

For every circuit skeleton $G$, the \emph{light cone size} of $G$ is the maximum over all output qubits $i$ of the size of the set $\{ j : \text{$j$ is input qubit connected to $i$ in $G$} \}$.
\end{definition}

\begin{figure}[H]
    \centering
    \scalebox{.8}{\input{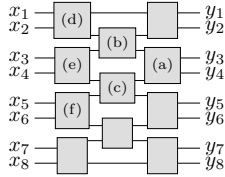}}
    \caption{An example of 1D circuit skeleton with $n=8$ and $d=3$. In this example the permutations are $\pi^{(0)}=\pi^{(4)}=\textsf{id},\ \pi^{(1)}=\pi^{(3)}=(18765432),\ \pi^{(2)}=(81234567)$.}
    \label{fig:circuit skeleton}
\end{figure}

Next, we define the light cone for an output qubit and the light cone size for a circuit skeleton.

\begin{definition}[Light cone]\label{def:light cone}
Let $\cS$ be a circuit skeleton and $i$ be an output qubit. The light cone of $i$ is the set of all input vertices in $\cS$ that has a path from left to right that ends at $i$. The light cone size of $\cS$ is then defined as the largest light cone size of an output qubit in $\cS$.
\end{definition}

Note that the light cone size of the 1D circuit in~\autoref{fig:circuit skeleton} is $6$, which is less than the number of qubits. Also, it turns out that computing the marginal of an output qubit only requires the information from the light cone.

\begin{lemma}[Marginal probability and light cone]\label{lem:marginal light cone}
Let $\cS$ be a circuit skeleton with light cone size $L$ and $C$ be a circuit using skeleton $\cS$. For each output qubit of $C$, the marginal probability can be computed in time $O(2^L)$.
\end{lemma}
\begin{proof}[Proof of~\autoref{lem:marginal light cone}]
We use the circuit skeleton in~\autoref{fig:circuit skeleton} as an illustrating example. For an output qubit in $C$, to compute its marginal probability it suffices to compute the input state to the gate it connects to. For example, for output qubit $y_3$, it suffices to compute the input state to gate (a).

Similarly, to compute the input state of a gate, it suffices to compute the input states of the gate it connects to from the previous layer. Namely, to compute the input state of gate (a), it suffices to compute that of gate (b) and (c). If we continue this process inductively, the only input state needed to compute the marginal probability of an output bit is then the one lies in its light cone. In this example, to compute the marginal probability of $y_3$, it suffices to consider only $x_1,x_2,\dots,x_6$.

Finally, to compute the input states of all the intermediate gates, it suffices to perform $2^L\times2^L$ matrix vector multiplication because each intermediate state is of size at most $L$. While all the above operations can be done in $O(2^L)$ times, computing the marginal probability of an output qubits in $C$ only requires $O(2^L)$ time.
\end{proof}

Now, we are able to formally define random quantum circuits.

\begin{definition}[Random quantum circuits]
Let $n\in2\N$, $d\in\N$. A distribution $\cD$ of $n$-qubit depth $d$ random circuits consists of an $n$-qubit depth $d$ circuit skeleton $\cS$ and ensembles $\{\cE_{i,j}\}$ over $4\times4$ unitary matrices for each $i=1,\dots,d$ and $j=1,\dots,n/2$.

A random quantum circuit $C$ sampled from $\cD$ by sampling a $4\times4$ unitary matrix $U_{i}^{(t)}$ from $\cE_{i}^{(t)}$ and assigning $U_{i}^{(t)}$ to the $i^\text{th}$ node of the $t^\text{th}$ layer for each $t=1,\dots,d$ and $i=1,\dots,n/2$.

Specifically, if each $\cE_{i}^{(t)}$ is Haar random, then we say $\cD$ is Haar random $2$-qubit circuits over $\cS$.
\end{definition}

\subsection{Tensor networks}\label{sec:prelim tensor net}

Tensor network is an intuitive graphical language that can be rigorously used in reasoning about multilinear maps. Especially, it finds many applications in quantum computing since the basic operations such as partial measurement are all multilinear maps. In this paper, we restrict our attention to \emph{qubits} (as opposed to the general case of \emph{qudits}) and only to gates that act on two qubits.

\begin{figure}[H]
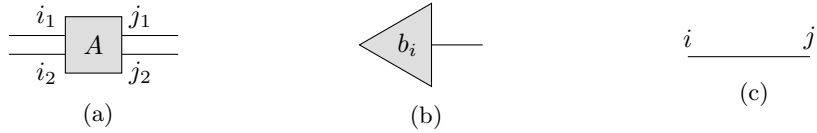

    \centering
    \begin{minipage}{0.2\linewidth}
            \centering
            \input{tikz/gate.tikz}
            \subcaption{}
            \label{fig:tensor net gate}
    \end{minipage}
    \hspace{0.05\linewidth}
    \begin{minipage}{0.2\linewidth}
            \centering
            \input{tikz/state.tikz}
            \subcaption{}
            \label{fig:tensor net state}
    \end{minipage}
    \hspace{0.05\linewidth}
    \begin{minipage}{0.2\linewidth}
            \centering
              \input{tikz/line.tikz}
              \subcaption{}
              \label{fig:tensor net line}
    \end{minipage}
    \caption{Three basic elements in tensor networks. (a) Gate: the figure represents $\sum_{b_i,b_j\in\{0,1\}}A_{b_i,b_j}\ket{b_i}\bra{b_j}$. (b) State: the figure represents $\bra{b_i}$. (c) Line: the figure represents $\delta_{b_i,b_j}$.}
    \label{fig:tensor net}
\end{figure}

In a tensor network, we represent a unitary matrix (\textit{e.g.,} a gate) as a box with lines on the sides (see~\autoref{fig:tensor net gate}). Each line represents a coordinate of the gate and in this paper each coordinate has dimension $2$ and is indexed by $\{0,1\}$. Specifically, a line on the left represents a column vector (\textit{i.e.,} $\ket{\cdot}$) while a line on the right  represents a row vector (\textit{i.e.,} $\bra{\cdot}$).\footnote{This is when the tensor network is written left to right - sometimes it is written top to bottom, in which case a line on the top represents a column vector and a line on the bottom represents a row vector.} For example,~\autoref{fig:tensor net gate} represents $\sum_{b_i,b_j\in\{0,1\}}A_{b_i,b_j}\ket{b_i}\bra{b_j}$.

Similarly, a state (\textit{e.g.,} a qubit)is represented by a triangle with line only on one side and it is a $\ket{\cdot}$ (resp. $\bra{\cdot}$) if the free-end of the line is left  (resp. right). For example,~\autoref{fig:tensor net state} represents $\bra{b_i}$.

Semantically, a pure line refers to an indicator function\footnote{Also known as \textit{contraction}.} for its two ends. For example, the line in~\autoref{fig:tensor net line} reads as $\sum_{b)i,b_j\in\{0,1\}}\delta_{b_i,b_j}\bra{b_i}\ket{b_j}$ where $\delta_{b_i,b_j}=1$ if $b_i=b_j$; otherwise it is $0$.

\section{Our Algorithm}

We now describe our classical algorithm that spoofs the linear cross-entropy benchmark in shallow quantum circuits.
The key idea is that rather than directly simulating the whole quantum circuit, our  algorithm only computes the marginal distributions of few output qubits and then samples substrings for those qubit accordingly. We sample the remaining subits uniformly at random. Intuitively, due to the correlation on those output qubits, one can expect that the linear cross-entropy of our algorithm could be better than uniform distribution, but  the analysis is somewhat delicate. Because consider shallow quantum circuits (of at most logarithmic light cone size), the marginal of few output qubits can be efficiently computed.

\begin{algorithm}[H] 
	\caption{Classical algorithm for spoofing linear XEB in shallow quantum circuits}\label{algo:main}
	\begin{algorithmic}[1]
	    \Input A quantum circuit $C$ sampled from $\cD$, a Haar random distribution over an $n$-qubit circuit skeleton $\cS$ with light cone size at most $L$.
	    \State We set $m \in \N$ to be a parameter in $\{1,\ldots,  \floor*{n/L} \}$. (We set $m=\floor*{n/L}$ to obtain the result of \autoref{thm:general} as stated.)
		\State Find $m$ output qubits $i_1,\dots,i_m$ such that their light cones are disjoint.
		\State Calculate the marginal probability of each output qubits $i_1,\dots,i_m$. 
		\State Sample $x_{i_1},\dots,x_{i_m}$ according to the marginal probabilities calculated in the previous step. For any $i\notin\{i_1,\dots,i_m\}$, sample $x_i$ uniformly random from $\{0,1\}$.
		\Output $x$.
	\end{algorithmic}
\end{algorithm}

\paragraph{Running time of the algorithm.}
The total running time of~\autoref{algo:main} is at most $\poly(n,2^L)$. 
Finding $m$ outputs with disjoint light cones takes  $\poly(n)$ time by a greedy algorithm.
The second step takes $\poly\left(n,2^L\right)$ time because it suffices to keep track of the $2^L\times 2^L$ density matrix recording the marginal probability of every qubit in the light cone of $i_j$ for each $j\in[m]$ (see \autoref{lem:marginal light cone}).
The final step of sampling uniform bits for the remaining outputs can be done in polynomial time.

\subsection{Analysis}

The following theorem implies  \autoref{thm:general} by setting $m=\floor*{n/L}$:

\begin{restatable}[Linear XEB for circuits with small light cones.]{theorem}{maintheoremrestate}\label{thm:generalrestate}
Let $n,d,L\in\N$ and let $\cD$ be a distribution over $n$-qubit quantum circuits with (i) light cone size at most $L$, (ii) depth at most $d$, and (iii) Haar random $2$-qubit gates. Then, letting $A_C$ be the distribution output by \autoref{algo:main} on input $C$, \[
\Exp_{C\sim\cD}\left[\cF_C(A_C)\right]\geq\left(1+15^{-d}\right)^m -1 \, \;,
\]
where $m$ is the parameter chosen in step 1 of the algorithm.
\end{restatable}

The proof of \autoref{thm:generalrestate} consists of three steps:

\begin{enumerate}
    \item  We reduce analyzing analyzing the expectation when the algorithms samples the marginals of $m$ output qubits into analyzing it for a single output qubit.
    
    \item We apply the integration formula for Haar measure and rewrite the expected linear XEB of a single qubit into a tensor network.
    
    \item We then perform a change of basis on the tensor network and turn the single qubit analysis into a Markov chain problem where the expected linear XEB of a single qubit can be easily lower bounded.
    
\end{enumerate}

Since the heart of the proof is the single output qubit analysis, we will describe it first.

\section{Single qubit analysis}\label{sec:single}

In this section, we prove the $m=1$ case of our algorithm.
That is, we prove that for a single  output qubit, the expected contribution to linear XEB is of the order of $15^{-d}$.

\begin{theorem}[Linear XEB of a single output qubit]\label{thm:single}
Let $n,d\in\N$ and $\cD$ be distribution over $n$-qubit quantum circuits with depth at most $d$ and with Haar random $2$-qubit gates. For $C\sim\cD$, let $U$ denote the unitary matrix computed by $C$. For each $i\in[n]$, we have
\[
\Exp_{C\sim\cD}\left[q_{C,i,0}^2+q_{C,i,1}^2\right]\geq\frac{1+15^{-d}}{2} \, ,
\]
where $q_{C,i,b} = \Pr_{x \sim q_C}[x_i =b]$.
\end{theorem}

We prove~\autoref{thm:single} by reducing to a Markov chain problem using tensor networks.
Without loss of generality we can assume $i=1$. Also, by~\autoref{lem:sum of squares of marginal}, 
\[
q_{C,1,0}^2+q_{C,1,1}^2=\frac{1+\tr\left(Z\otimes I^{\otimes n-1}U^\dagger\ket{0^n}\bra{0^n}U\right)^2}{2}
\]
So our goal is to show that
\begin{equation}\label{eq:trace}
\Exp_{C\sim\cD}\left[\tr\left(Z\otimes I^{\otimes n-1}U^\dagger\ket{0^n}\bra{0^n}U\right)^2\right]\geq15^{-d} \ \ \ \  .
\end{equation}

Let us start with rewriting the trace term of~\autoref{eq:trace} into an equivalent tensor network expression as follows.

\[
\tr\left(Z\otimes I^{\otimes n-1}U^\dagger\ket{0^n}\bra{0^n}U\right)^2=\scalebox{.8}{\input{tikz/trace-circuit.tikz}} \, .
\]

Next, for a single gate $g$ in a quantum circuit, its expected behavior over the choice of $2$-qubit Haar random gates can be characterized in the following lemma.

\begin{lemma}\label{lem:single}
Let $U_g$ be Haar random $2$-qubit gate, then the following holds.
\[
\Exp_{U_g}\left[\scalebox{.8}{\input{tikz/transition-matrix-1.tikz}}\right]=\sum_{\sigma_1,\sigma_2,\sigma_1',\sigma_2'\in\{I,X,Y,Z\}}\scalebox{.8}{\input{tikz/transition-matrix.tikz}} \ \ \  .
\]
where
\[
M_{\sigma_1,\sigma_2,\sigma_1',\sigma_2'} =\begin{blockarray}{cccccc}
II & IX & IY & \cdots & ZZ \\
\begin{block}{(ccccc)c}
  1 & 0 & 0 & 0 & 0 & II \\
  0 & \frac{1}{15} & \frac{1}{15} & \cdots & \frac{1}{15} & IX \\
  0 & \frac{1}{15} & \frac{1}{15} & \cdots & \frac{1}{15} & IY \\
  0 & \vdots & \vdots & \ddots & \vdots & \vdots \\
  0 & \frac{1}{15} & \frac{1}{15} & \cdots & \frac{1}{15} & ZZ \\
\end{block}
\end{blockarray} \, .
 \]
\end{lemma}

The proof of~\autoref{lem:single} is based on the integration formula~\cite{BB96} for Haar measure.
We postpone the proof of~\autoref{lem:single} to~\autoref{sec:single lemma}. Intuitively, the lemma says that by a change a basis, the expected behavior of a single Haar random $2$-qubit gate can be exactly understood by an explicit transition matrix $M$. By the linearity of taking expectation, we can apply~\autoref{lem:single} on every gates in the circuit $C$ and thus the whole tensor network is simplified to a \textit{Markov chain}. Concretely, we have the following lemma.

\begin{lemma}[Rewrite~\autoref{eq:trace} as a Markov chain]\label{lem:markov chain}
Let $n,d\in\N$ and $\cD$ be a Haar random distribution over an $n$-qubit depth $d$ circuit skeleton $\cS$ with permutations $\pi^{(0)},\pi^{(1)},\dots,\pi^{(d)}$. For $C\sim\cD$, let $U$ denote the unitary matrix computed by $C$.
\begin{equation}\label{eq:markov chain}
\Exp_{C\sim\cD}\left[\tr\left(Z\otimes I^{\otimes n-1}U^\dagger\ket{0^n}\bra{0^n}U\right)^2\right]=\ \sum_{\substack{\sigma_{i'}^{(t')}\in\{I,X,Y,Z\}\\i'=1,2,\dots,n\\t'=1,2,\dots,d+1}}\prod_{t=0}^{d+1} V^{(t)}\left(\left\{\sigma_{i'}^{(t')}\right\}\right)
\end{equation}
where
\[
V^{(t)}\left(\left\{\sigma_{i'}^{(t')}\right\}\right)=\left\{\begin{array}{ll}
\prod_{i=1}^n\tr\left(\frac{I+Z}{2}\frac{\sigma_{i}^{(1)}}{\sqrt{2}}\right)^2     & ,\ \text{if }t=0 \\
\prod_{i=1}^{n/2}M_{\sigma_{2i-1}^{(t)},\sigma_{2i}^{(t)},\sigma_{\pi^{(t)}(2i-1)}^{(t+1)},\sigma_{\pi^{(t)}(2i)}^{(t+1)}}     & ,\ \text{if }t=1,2,\dots,d\\
\tr\left(\frac{\sigma_{1}^{(d+1)}}{\sqrt{2}}Z\right)^2\cdot\prod_{i=2}^n\tr\left(\frac{\sigma_{i}^{(d+1)}}{\sqrt{2}}I\right)^2     & ,\ \text{if }t=d+1 \, .
\end{array} \right.
\]
\end{lemma}

The proof of~\autoref{lem:markov chain} is based on a careful composition of applying~\autoref{lem:single} on each of the gates. We postpone the proof of~\autoref{lem:markov chain} to~\autoref{sec:markov chain}. Now, we are ready to prove~\autoref{thm:single} and complete the analysis for the expected linear XEB of single output qubit.

\begin{proof}[Proof of~\autoref{thm:single}]
\autoref{lem:markov chain} rewrites the desiring quantity into the form of a Markov chain so now it suffices to show that the right hand side of~\autoref{eq:markov chain} is at least $15^{-d}$.

Notice that for every possible assignment to $\{\sigma_{i'}^{(t')}\}$, $V(\{\sigma_{i'}^{(t')}\})\geq0$. That is, it suffices to find an assignment such that $\prod_{t=0}^{d+1}V^{(t)}(\{\sigma_{i'}^{(t')}\})\geq15^{-d}$. Specifically, let us consider the following assignment. For all $i=1,2,\dots,n$ and $t=1,2,\dots,d+1$, let
\[
\sigma_i^{(t)} = \left\{\begin{array}{ll}
Z     & ,\ \text{if } \pi^{(d)}\circ\pi^{(d-1)}\circ\cdots\circ\pi^{(t)}(i)=1 \\
I     & ,\ \text{else} \, .
\end{array} \right.
\]
To analyze this assignment, let us start with the last layer. There we have $\sigma_1^{(d+1)}=Z$ and $\sigma_i^{(d+1)}=I$ for each $i=2,3,\dots,n$ and thus
\begin{align*}
V^{(d+1)}(\{\sigma_{i'}^{(t')}\})&=\tr\left(\frac{\sigma_{1}^{(d+1)}}{\sqrt{2}}Z\right)^2\cdot\prod_{i=2}^n\tr\left(\frac{\sigma_{i}^{(d+1)}}{\sqrt{2}}I\right)^2\\
&=\tr\left(\frac{ZZ}{\sqrt{2}}\right)^2\cdot\tr\left(\frac{II}{\sqrt{2}}\right)^{2(n-1)}=2^n \, .
\end{align*}
Next, for each $t=1,2,\dots,d$ and $i=1,2,\dots,n$, observe that $\sigma_i^{(t)}=\sigma_{\pi^{(t+1)}(i)}^{(t)}$ due to the choice of the assignment. As a result, all the $M_{\sigma_{2i-1}^{(t)},\sigma_{2i}^{(t)},\sigma_{\pi^{(t)}(2i-1)}^{(t+1)},\sigma_{\pi^{(t)}(2i)}^{(t+1)}}$ will be either $M_{I,I,I,I}$ or $M_{I,Z,I,Z}$. Specifically, for each $t=1,2,\dots,d$, since there is exactly one $Z$ appears among $\{\sigma_i^{(t)}\}_{i=1,\dots,n}$ while the rest are $I$s, there is also exactly one $M_{I,Z,I,Z}$ term contributes in $V^{(t)}(\{\sigma_{i'}^{(t')}\})$ while the other terms are $M_{I,I,I,I}$. Namely, we have
\[
V^{(t)}(\{\sigma_{i'}^{(t')}\}) = M_{I,Z,I,Z}\cdot (M_{I,I,I,I})^{n/2-1} = \frac{1}{15}
\]
for each $t=1,2,\dots,d$.

Finally, since there is exactly one $Z$ appears in $\{\sigma_i^{(0)}\}_{i=1,\dots,n}$ while the rest are $I$s, we have
\begin{align*}
V^{(0)}(\{\sigma_{i'}^{(t')}\}) &=\prod_{i=1}^n\tr\left(\frac{I+Z}{2}\frac{\sigma_{i}^{(0)}}{\sqrt{2}}\right)^2= \tr\left(\frac{I+Z}{2}\frac{Z}{\sqrt{2}}\right)^2\cdot\tr\left(\frac{I+Z}{2}\frac{I}{\sqrt{2}}\right)^{2(n-1)}\\
&=\left(\frac{1}{\sqrt{2}}\right)\cdot\left(\frac{1}{\sqrt{2}}\right)^{2(n-1)}=\frac{1}{2^n} \, .
\end{align*}
To sum up, we conclude that $V(\{\sigma_{i'}^{(t')}\})=\prod_{t=0}^{d+1}V^{(t)}(\{\sigma_{i'}^{(t')}\})=15^{-d}$ as desired. Specifically, this implies~\autoref{eq:trace}, \textit{i.e.,} $\Exp_{C\sim\cD}\left[\tr\left(Z\otimes I^{\otimes n-1}U^\dagger\ket{0^n}\bra{0^n}U\right)^2\right]\geq15^{-d}$. Combine with~\autoref{lem:sum of squares of marginal}, this completes the proof of~\autoref{thm:single}.
\end{proof}

\subsection{Proof of Lemma~\ref{lem:single}}\label{sec:single lemma}
We start with applying the integration formula for Haar random matrix and considering its tensor netowrok representation.
\begin{lemma}[{\cite[Equation~2.4]{BB96}}]
Let $U$ be a Haar random $2$-qubit gate, then we have the following. For each $x_a,x_b,x_c,x_d,y_a,y_b,y_c,y_d\in\{0,1\}^2$,
\begin{align*}
    \Exp_U\left[U_{x_ay_a}U^\dagger_{x_by_b}U_{x_cy_c}U^\dagger_{x_dy_d}\right] &= \frac{1}{15}\cdot\Big[\delta_{x_ax_b}\delta_{x_cx_d}\delta_{y_ay_b}\delta_{y_cy_d}+\delta_{x_ax_d}\delta_{x_bx_c}\delta_{y_ay_d}\delta_{y_by_c}\Big]\\
    &-\frac{1}{60}\cdot\Big[\delta_{x_ax_b}\delta_{x_cx_d}\delta_{y_ay_d}\delta_{y_by_c}+\delta_{x_ax_d}\delta_{x_bx_c}\delta_{y_ay_b}\delta_{y_cy_d}\Big] \, .
\end{align*}
The above equation can be represented as the following tensor network.
\[
\Exp_{U_g}\left[\scalebox{.8}{\input{tikz/transition-matrix-1.tikz}}\right]=\frac{1}{15}\cdot\left[\scalebox{.8}{\input{tikz/single-exp-1.tikz}}\right]-\frac{1}{60}\cdot\left[\scalebox{.8}{\input{tikz/single-exp-2.tikz}}\right] \ \ .
\]
\end{lemma}

Next, the idea is to apply the Pauli identity (\textit{i.e.,}~\autoref{eq:pauli identity}) on each pair of $(1_a,1_b)$,$(1_c,1_d)$, $(2_a,2_b)$, $(2_c,2_d)$, $(1_a',1_b')$, $(1_c',1_d')$, $(2_a',2_b')$, and $(2_c',2_d')$. Intuitively, this is doing a change of basis from the standard basis to Pauli basis.

Let us first apply the Pauli identity on $(1_a,1_b)$ and $(1_c,1_d)$ note that by~\autoref{lem:pauli tensor net}, we have

\begin{align}
\scalebox{.8}{\input{tikz/pauli-one-circle.tikz}}&=\left\{\begin{array}{ll}
2\ \scalebox{.7}{\input{tikz/pauli-one-gate.tikz}}     & ,\ \sigma=\sigma'  \\
0     & ,\ \text{else}
\end{array}
\right.\label{eq:tensor net pauli one circule}
\intertext{and}
\scalebox{.8}{\input{tikz/pauli-two-circles.tikz}}&=\left\{\begin{array}{ll}
4\ \scalebox{.7}{\input{tikz/pauli-one-gate.tikz}}     & ,\ \sigma=\sigma'=I  \\
0     & ,\ \text{else.}
\end{array}
\right.\label{eq:tensor net pauli two circles}
\end{align}

That is, the tensor network is non-zero only if $\sigma=\sigma'$. Thus, we only need one variable $\sigma_1\in\{I,X,Y,Z\}$ to handle $(1_a,1_b)$ and $(1_c,1_d)$. Similarly, we can use $\sigma_2,\widetilde{\sigma_{1}},\widetilde{\sigma_{2}}\in\{I,X,Y,Z\}$ to handle other pairs respectively. The equation becomes the following.
\[
\Exp_{U_g}\left[\scalebox{.5}{\input{tikz/transition-matrix-1.tikz}}\right]=\frac{1}{2^8}\sum_{\substack{\sigma_1',\sigma_2',\sigma_1'',\sigma_2''\\\in\{I,X,Y,Z\}}}\scalebox{.5}{\input{tikz/single-exp-3.tikz}}\left(\frac{1}{15}\cdot\left[\scalebox{.5}{\input{tikz/single-exp-5.tikz}}\right]-\frac{1}{60}\cdot\left[\scalebox{.5}{\input{tikz/single-exp-6.tikz}}\right]\right)\scalebox{.5}{\input{tikz/single-exp-7.tikz}} \ .
\]

To finish the proof of~\autoref{lem:single}, we have to explicitly calculate the value of the tensor network for each choice of $\sigma_1,\sigma_2,\sigma_{1'},\sigma_{2'}\in\{I,X,Y,Z\}$. Again, by~\autoref{eq:tensor net pauli one circule} and~\autoref{eq:tensor net pauli two circles}, we have the following observations.
\begin{itemize}
\item If $\sigma_1=\sigma_2=\sigma_{1'}=\sigma_{2'}=I$, then the value is $2^{-8}\cdot\left(15^{-1}\cdot(2^8+2^4)-60^{-1}\cdot(2^6+2^6)\right)=2^{-4}$.
\item If $\sigma_1=\sigma_2=I$ and at least one of $\sigma_{1'},\sigma_{2'}$ is not $I$, or at least one of $\sigma_1,\sigma_2$ is not $I$ and $\sigma_{1'}=\sigma_{2'}=I$, then the value is $2^{-4}\cdot(15^{-1}\cdot(0+2^4)+60^{-1}\cdot(2^6+0))=0$.
\item For all the other cases, the value is $2^{-4}\cdot(15^{-1}\cdot(0+2^4)+60^{-1}\cdot(0+0))=2^{-4}\cdot15^{-1}$.
\end{itemize}
Finally, we take out the $2^{-4}$ and evenly distribute it to the Pauli gates outside. Namely, each of them gets an extra $1/\sqrt{2}$ factor as shown in the equation. This completes the proof of~\autoref{lem:single}.

\subsection{Proof of Lemma~\ref{lem:markov chain}}\label{sec:markov chain}

Let us do a change of basis from the standard basis to the Pauli basis. Concretely, we apply~\autoref{lem:single} on every gate. Note that by the independence of each gate and the linearity of expectation, the $t^\text{th}$ layer of the circuit becomes the following for each $t=1,2,\dots,d$.

\[
\Exp_{U_t}\scalebox{.8}{\input{tikz/markov-layer-1.tikz}}=\sum_{\substack{\sigma_i^{(t)},\widetilde{\sigma}_i^{(t)}\in\{I,X,Y,Z\}\\\forall i=1,2,\dots,n}}\scalebox{.6}{\input{tikz/markov-layer-2.tikz}}\prod_{i=1}^{n/2}M_{\sigma_{2i-1}^{(t)},\sigma_{2i}^{(t)},\widetilde{\sigma}_{2i-1}^{(t)},\widetilde{\sigma}_{2i}^{(t)}}\scalebox{.6}{\input{tikz/markov-layer-3.tikz}} \ .
\]

Next, the $i^\text{th}$ output wire at the $t^\text{th}$ layer, \textit{i.e.,} the wires indexed by $\widetilde{i_a},\widetilde{i_b},\widetilde{i_c},\widetilde{i_d}$, connects to the $(\pi^{(t)}(i))^\text{th}$ input wire at the $(t+1)^\text{th}$ later, \textit{i.e.,} the wires indexed by $\pi^{(t)}(i)_a,\pi^{(t)}(i)_b,\pi^{(t)}(i)_c,\pi^{(t)}(i)_d$. By the orthogonality of Pauli gates (\textit{i.e.,}~\autoref{lem:pauli tensor net}), we have $\tr(\frac{\widetilde{\sigma}_i^{(t)}}{\sqrt{2}}\frac{\sigma_{\pi^{(t)}(i)}^{(t+1)}}{\sqrt{2}})=\delta_{\widetilde{\sigma}_i^{(t)},\sigma_{\pi^{(t)}(i)}^{(t+1)}}$ for all $i=1,2,\dots,n$. To sum up, the $1^\text{st}$ to $d^\text{th}$ layer is equivalent to following.

\[
\Exp_{U_1,U_2,\dots,U_d}\left[\scalebox{.5}{\input{tikz/exp-circuit.tikz}}\right] = \sum_{\substack{\sigma_i^{(t)}\in\{I,X,Y,Z\}\\\forall i=1,2,\dots,n\\t=1,2,\dots,t+1}}\scalebox{.45}{\input{tikz/exp-circuit-1.tikz}}\prod_{t=1}^d\left(\prod_{i=1}^{n/2}M_{\sigma_{2i-1}^{(t)},\sigma_{2i}^{(t)},\sigma_{\pi^{(t)}(2i-1)}^{(t+1)},\sigma_{\pi^{(t)}(2i)}^{(t+1)}}\right)\scalebox{.45}{\input{tikz/exp-circuit-2.tikz}} \ .
\]

Finally, let us plug in the input and output layer. Recall that the input layer contains 4 copies of $\ket{0^{n}}$ and the output layer contains 2 copies of $Z\otimes I^{\otimes n-1}$. Concretely, the contribution from the input layer would be
\[
\prod_{i=1}^n\left(\bra{0}\frac{\sigma_1^{(1)}}{\sqrt{2}}\ket{1}\right)^2=\prod_{i=1}^n\tr\left(\frac{I+Z}{2}\frac{\sigma_i^{(1)}}{\sqrt{2}}\right)^2
\]
while the contribution from the output layer would be
\[
\tr\left(\frac{\sigma_1^{(d+1)}}{\sqrt{2}}Z\right)^2\cdot\prod_{i=2}^n\tr\left(\frac{\sigma_i^{(d+1)}}{\sqrt{2}}I\right)^2 \, .
\]

This completes the proof of~\autoref{lem:markov chain}.

\section{Wrapping up: from single output bit to many bits}\label{sec:exp}

In this section we complete the proof of \autoref{thm:general} .

We will use the following notation. Let $q(x)$ be a pdf over $x\in\{0,1\}^n$. For any $I\subset[n]$ and $x_I\in\{0,1\}^I$, let $q(I,x_I)$ denote the marginal probability of the output qubit at location $I$ being $x_I$. Formally, $q(I,x_I) = \sum_{\substack{y\in\{0,1\}^n\\y_I=x_I}}q(y)$. 
For a fixed input $C$, let $I=\{i_1,\dots,i_m\}$ be the output qubits selected by~\autoref{algo:main}. Note that~\autoref{algo:main} will choose the same $I$ for each $C$ sampled from $\cD$. By the design of~\autoref{algo:main}, $A_C(x) = \frac{1}{2^{n-m}}q_C(I,x_I)$
for every $x\in\{0,1\}^n$. Thus, the linear XEB of $A_C$ is the following.
\begin{align}
\cF_C(A_C)&=2^n\sum_{x\in\{0,1\}^n}q_C(x)A_C(x)-1=2^n\sum_{x\in\{0,1\}^n}q_C(x)\frac{q_C(I,x_I)}{2^{n-m}}-1\\
&=2^m\sum_{x_I\in\{0,1\}^I}q_C(I,x_I)^2-1 \, . \nonumber
\intertext{Note that because their light cones are disjoint, by~\autoref{lem:disjoint light cone}, we have $q_C(I,x_I) = \prod_{j=1}^mq_C\left(\{i_j\},x_{i_j}\right)$. Thus, the equation becomes}
&=\sum_{x_I\in\{0,1\}^I}\prod_{j=1}^m2q_C(\{x_{i_j}\},x_{i_j})^2-1 \, . \label{eq:XEB pc}
\end{align}

Now, let us take expectation on the linear XEB over $\cD$. Since $\cD$ fixes the structure of the circuit and the randomness only lies in the choice of gates, $q_C\left(\{i_j\},x_{i_j}\right)$ is independent to each other. Namely,
\begin{align}
\Exp_{C\sim\cD}\left[\cF_C(A_C)\right]&=\sum_{x_I\in\{0,1\}^I}\prod_{j=1}^m2\Exp_{C\sim\cD}\left[q_C(\{x_{i_j}\},x_{i_j})^2\right]-1\nonumber\\
&=\prod_{j=1}^m2\Exp_{C\sim\cD}\left[q_C(\{x_{i_j}\},0)^2+q_C(\{x_{i_j}\},1)^2\right]-1 \, .\label{eq:exp XEB}
\end{align}

Using the single qubit analysis (\autoref{thm:single}), we can complete the proof of~\autoref{thm:general} as follows.

\begin{proof}[Proof of~\autoref{thm:general}]
Apply~\autoref{thm:single} on~\autoref{eq:exp XEB}, we have
\[
\Exp_{C\sim\cD}\left[\cF_C(A_C)\right]\geq\left(1+15^{-d}\right)^m-1
\]
as desired.
\end{proof}

\subsection{Probability over circuits} \label{sec:probcircuits}

Using~\autoref{thm:general}, we can obtain the following lower bound on the probability over the choice of the circuit $C$ of obtaining non-trivial fidelity:

\begin{restatable}[Lower bounding for the probability of success]{corollary}{generalprob} \label{cor:generalprob}
Let $n,d,\cD,L$ be as in Theorem~\ref{thm:general}. Then there is a randomized $poly(n,2^L)$ time algorithm such that for every  $1 \leq n \leq \floor*{\frac{n}{L}}$ and $0<\epsilon<1$, 
\[
\Pr_{C\sim\cD}\left[\cF_C\left(A_C\right)\geq\left(\left(1+15^{-d}\right)^{m}-1\right)\right] \geq \frac{(1-\epsilon)\cdot\left(\left(1+15^{-d}\right)^m-1\right)}{2^m-1} \geq \Omega\left( \frac{m \cdot 15^{-d}}{2^m} \right)\, .
\]
\end{restatable}

\begin{proof}[Proof of~\autoref{cor:generalprob}]
The idea is simple - since our algorithm picks $n-m$ bits uniformly at random, for every circuit $C$, by~\autoref{eq:XEB pc},  $\cF_C(A_C)\leq2^m$. Now, for any $0<\epsilon<1$, let $\delta=\Pr_{C\sim\cD}\left[\cF_C(A_C)>\epsilon\cdot\left(\left(1+15^{-d}\right)^m-1\right)\right]$, we have
\[
\left(1+15^{-d}\right)^m-1\leq\Exp_{C\sim\cD}\left[\cF_C(A_C)\right]\leq\delta\cdot2^m+\epsilon\cdot\left(\left(1+15^{-d}\right)^m-1\right) \, .
\]
Thus,
\[
\delta\geq\frac{(1-\epsilon)\cdot\left(\left(1+15^{-d}\right)^m-1\right)}{2^m} \, .
\]

\end{proof}

\section{Sample complexity analysis}\label{sec:sample complexity}

In this section, we discuss the empirical linear XEB of our algorithm. Namely, how many samples are required so that the empirical average of the linear XEB can be non-trivially lower bounded. Specifically, the goal would be the following. For some $T=\poly(n)$,
\begin{equation}\label{eq:sample complexity goal}
\Pr_{\substack{C\sim\cD\\x_1,\dots,x_T\sim A_C}}\left[\frac{1}{T}\sum_{i=1}^n\cF_C(x_i)=\Omega(1)\right]\geq\frac{1}{\poly(n)} \, .
\end{equation}

In~\autoref{sec:exp}, we have shown that the expectation of the linear XEB of our algorithm is at least $\left(1+15^{-d}\right)^{m}$ for $1\leq m\leq\floor*{n/L}$ with probability $1/\poly(n)$ over the choice of random circuits. Thus, to achieve~\autoref{eq:sample complexity goal}, it suffices to show that the probability of the empirical average of the linear XEB deviating from $\cF_C(x)$ is small.

In general, it is a difficult task to rigorously upper bound the sample complexity of linear XEB.
The reason is that such analysis needs to handle higher moment of $q_C$ which is highly non-trivial for even 2D circuits.
In this work we stick with the simpler case of analyzing the variance of linear XEB in~\autoref{lem:sample complexity and variance}. We further show in~\autoref{lem:variance to collision prob} that an inverse exponential bound on the collision probability of $q_C$ would be sufficient for giving $\poly(n)$ upper bound for the sample complexity.

\subsection{A variance/collision probability approach}
The variance of $\cF_C(x)$ is sufficient for upper bounding the number of samples required for the empirical linear XEB to converge. Specifically,Chebyshev's inequality implies that with $\Var_{x\sim p}[\cF_C(x)]/(\epsilon^2\delta)$ many samples, the empirical XEB is at least $\cF_C(p)-\epsilon$ with probability $\delta$ over the randomness of $p$. Note that here the circuit $C$ is fixed.

\begin{lemma}\label{lem:sample complexity and variance}
Let $C$ be an $n$-qubit quantum circuit and $q_C$ be the pdf of the distribution obtained from $C\ket{0^n}$. For any pdf $p:\{0,1\}^n\rightarrow[0,1]$ and $\epsilon,\delta\in(0,1)$, we have
\[
\Pr_{x_1,\dots,x_T\sim p}\left[\frac{1}{T}\sum_{i=1}^T\cF_C(x_i)\leq\cF_C(p)-\epsilon\right]\leq\delta
\]
when $T\geq\frac{\Var_{x\sim p}[\cF_C(x)]}{\epsilon^2\delta}$.
\end{lemma}
\begin{proof}
Since $\{\cF_C(x_i)\}$ are i.i.d. random variables with mean $\cF_C(p)$ and variance $\Var_{x\sim p}[\cF_C(x)]$, by Chebyshev's inequality, we have
\begin{align*}
\Pr_{x_1,\dots,x_T\sim p}\left[\frac{1}{T}\sum_{i=1}^T\cF_C(x_i)<\cF_C(p)-\epsilon\right] &\leq\frac{\Var_{x\sim p}[\cF_C(x)]}{T\cdot\epsilon^2} \, .
\end{align*}
As we pick $T\geq\frac{\Var_{x\sim p}[\cF_C(x)]}{\epsilon^2\delta}$, the above error is at most $\delta$ desired.
\end{proof}

The following lemma further shows that to upper bound the variance of our algorithm, it suffices to bound the \textit{collision probability} of the ideal distribution.

\begin{lemma}\label{lem:variance to collision prob}
Let $n,d,L\in\N$ and $\cD$ be distribution over $n$-qubit quantum circuits with (i) light cone size at most $L$, (ii) depth at most $d$, and (iii) with Haar random $2$-qubit gates. Let $1\leq m\leq\floor*{n/L}$ and $A$ be the algorithm from~\autoref{algo:main}, we have
\[
\Var_{x\sim A_C}[\cF_C(x)]\leq2^{m+n}\sum_{x\in\{0,1\}^n}q_C(x)^2
\]
where $\sum_{x\in\{0,1\}^n}q_C(x)^2$ is also known as the collision probability of $q_C$.
\end{lemma}
\begin{proof}[Proof of~\autoref{lem:variance to collision prob}]
Consider the variance of the linear XEB of our algorithm $A_C$ as follows.
\begin{align*}
\Var_{x\sim A_C}[\cF_C(x)]&\leq\Exp_{x\sim A_C}[2^{2n}q_C(x)^2]=2^{2n}\sum_{x\in\{0,1\}^n}A_C(x)q_C(x)^2 \, .
\intertext{Recall that $A_C(x)\leq2^{m-n}$ for all $x$, thus the equation becomes}
&\leq2^{m+n}\sum_{x\in\{0,1\}^n}q_C(x)^2 \, .
\end{align*}
\end{proof}

To have some intuition on~\autoref{lem:variance to collision prob}, the right hand side is minimized when $q_C$ is the uniform distribution over $\{0,1\}^n$ where the collision probability is $2^{-n}$. In such case, the variance of our algorithm is $O(2^m)$. When choosing $m=O(\log n)$, the sample complexity of our algorithm would be $\poly(n)$ as desired.\\

In general, using the variance/collision probability to upper bound the sample complexity might not be tight. For example, consider the distribution of a sequence of independent biased coins, \textit{i.e.,} $q(x)=(1/2+\epsilon)^{\|x\|_1}\cdot(1/2-\epsilon)^{n-\|x\|_1}$ for each $x\in\{0,1\}^n$. Then the variance $\Var_{x\sim q}[q(x)]$ is exponentially large, however, the sample complexity of having the empirical average of $q(x)$ being of the order of $\Exp_{x\sim q}[q(x)]$ is $O(1)$ with high probability. Specifically, when the depth of the random circuit is $1$, then the marginal distribution looks like the above biased coins distribution with high probability and thus undesirable.

On the other extreme where the random circuit is very deep, it is known that the marginal distribution will converge to the \textit{Porter Thomas distribution} and its collision probability is $O(2^{-n})$ \cite{google19}.

In~\autoref{sec:sample complexity 1D}, we further show that the collision probability of $q_C$ is $O(2^{-n})$ in expectation for 1D circuit of depth at least $(\log n)/\log(5/4)$.
While the proof could potentially be extended to 2D circuit and beyond, we leave it as a future direction and state the following conjecture.

\begin{restatable}{conjecture}{collisionconjecture}\label{conj:collision 2D}
Let $n,d\in\N$ and $\cD$ be a distribution of $n$-qubit. For a circuit $C$, denote $q_C$ as the pdf of $C\ket{0^n}$. We conjecture that there exists a constant $c>0$ such that when $\cD$ is the distribution over 2D random circuits of depth $d\geq c\sqrt{\log n}$,
\[
\Exp_{C\sim\cD}\left[\sum_{x\in\{0,1\}^n}q_C(x)^2\right]=O\left(\frac{1}{2^n}\right) \, .
\]
\end{restatable}

\subsection{The sample complexity of 1D random circuits of logarithmic depth}\label{sec:sample complexity 1D}

In this subsection, we formally prove that the sample complexity of our algorithm is $O\left(2^m\right)$ for random 1D circuits with high probability.

\begin{theorem}\label{thm:sample complexity 1D}
Let $n\in\N$ and $\cD$ be the distribution over $n$-qubit 1D quantum circuits with depth $d=\Omega(\log n)$ and with Haar random $2$-qudit gates where the dimension of the qudit is at least $4$. Let $1\leq m\leq\floor*{n/2d}$ be the number of output qubits used by our algorithm. Then for any $\delta\in(0,1)$, we have
\[
\Pr_{C\sim\cD}\left[\Var_{x\sim A_C}[\cF_C(x)]=O\left(\frac{2^m}{\delta}\right)\right]\geq1-\delta \, .
\]
Specifically, combine with~\autoref{thm:general}, we have
\[
\Pr_{\substack{C\sim\cD\\x_1,\dots,x_T\sim A_C}}\left[\sum_{i=1}^T\cF_C(x_i)=\Omega\left(\frac{1}{\poly(n)}\right)\right]\geq\frac{1}{\poly(n)}
\]
when $T=\Omega(1)$.
\end{theorem}

\begin{proof}[Proof of~\autoref{thm:sample complexity 1D}]

Let us first show that upper bounding the second moment of $q_C$ is sufficient for proving~\autoref{thm:sample complexity 1D}. Consider the variance of the linear XEB of our algorithm $A_C$ as follows.
\begin{align*}
\Var_{x\sim A_C}[\cF_C(x)]&\leq\Exp_{x\sim A_C}[2^{2n}q_C(x)^2]=2^{2n}\sum_{x\in\{0,1\}^n}A_C(x)q_C(x)^2 \, .
\intertext{Recall that $A_C(x)\leq2^{m-n}$ for all $x$, thus the equation becomes}
&\leq2^{m+n}\sum_{x\in\{0,1\}^n}q_C(x)^2 \, .
\end{align*}

Next, the lemma below shows that the second moment term $\sum_{x\in\{0,1\}^n}q_C(x)^2$ is exponentially small with high probability over the choice of $C$.

\begin{restatable}{lemma}{anticoncentrationoneD}
\label{lem:anti concentration of 1D}
Let $n\in\N$ and $\cD$ be the distribution over $n$-qubit 1D quantum circuits with depth at least $\frac{\log n}{\log(5/4)}$ and with Haar random $2$-qubit gates. Then we have
\[
\Exp_{C\sim\cD}\left[\sum_{x\in\{0,1\}^n}q_C(x)^2\right]=O\left(\frac{1}{2^n}\right) \, .
\]
\end{restatable}
The proof of~\autoref{lem:anti concentration of 1D} is based on the Ising model analysis by~\cite{Hunter19}. We postpone it to~\autoref{sec:proof of anti concentration 1D}. Now, let us complete the proof of~\autoref{thm:sample complexity 1D}. By~\autoref{lem:anti concentration of 1D}, we have
\begin{align*}
\Exp_{C\sim\cD}\left[\Var_{x\sim A_C}[\cF_C(x)]\right]&\leq 2^{m+n}\Exp_{C\sim\cD}\left[q_C(x)^2\right]\leq O\left(2^m\right) \, .
\end{align*}
Thus, for any $\delta>0$, by Markov's inequality, we have $\Pr_{C\sim\cD}[\Var_{x\sim A_C}[\cF_C(x)]=O(2^m/\delta)]\geq1-\delta$ as desired.
\end{proof}

\subsection{Proof of Lemma~\ref{lem:anti concentration of 1D}}\label{sec:proof of anti concentration 1D}

It turns out that the previous Markov chain approach in analysis the expected linear XEB of our algorithm is not sufficient for upper bounding the expectation of $\sum_{x\in\{0,1\}^n}q_C(x)^2$. We thus consider a different approach by reducing the quantity to a combinatorial problem in a spin system on lattice. The proof is highly inspired by a recent paper of Hunter~\cite{Hunter19}.

For the convenience of the analysis, here we fix the following skeleton for 1D circuit while the result can be easily extended to other variants.

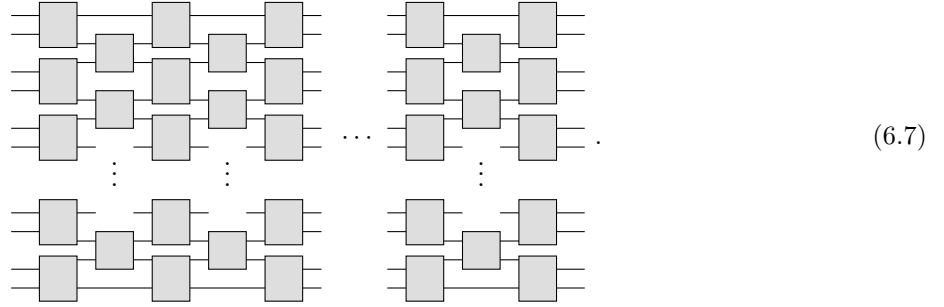
\begin{equation}\label{eq:1D}
\input{tikz/1D.tikz} \ .
\end{equation}

\paragraph{Step 1: Reducing to counting spin configurations on a hexagonal lattice}
First, let us rewrite $\sum_{x\in\{0,1\}^n}q_C(x)^2$ into an equivalent tensor network.
\begin{align*}
\sum_{x\in\{0,1\}^n}q_C(x)^2 &= \sum_{x\in\{0,1\}^n}\tr\left(\bra{0^n}U\ket{x}\bra{x}U^\dagger\ket{0^n}\right)^2\\
&=\sum_{x\in\{0,1\}^n}\scalebox{.8}{\input{tikz/second-moment-tensornet.tikz}} \ .
\end{align*}

In~\autoref{lem:single}, we change the basis to the Pauli basis and replace the expectation of a single gate with a transition matrix. Here, we instead stick to the permutation basis in the integration formula and represent a single gate as an \textit{effective vertex}~\cite{BB96}.

\begin{lemma}[\cite{BB96}]\label{lem:single exp permutation}
Let $U_g$ be a Haar-random $2$-qubit gate. We have the following.
\[
\Exp_{U_g}\input{tikz/transition-matrix-1.tikz} = \sum_{\sigma,\tau\in S_2} \input{tikz/single-exp-4.tikz} \ .
\]
Here $S_2=\{\mathbb{I},\mathbb{S}\}$ is the permutation group of two elements and an edge on the left represents four edges on the right. Specifically, the edge with $\sigma$ and $\tau$ on the two ends has weight $\bra{\sigma}\ket{\tau}$ where
\[
\bra{\sigma}\ket{\tau} = \left\{\begin{array}{ll}
\frac{1}{15}     & ,\ \sigma=\tau \\
\frac{-1}{60}     & ,\ \sigma\neq\tau \, .
\end{array}
\right.
\]
for all $\sigma,\tau\in S_2$. As for the boundary condition, for each $b_i\in\{0,1\}$, we have
\[
\scalebox{.8}{\input{tikz/spin-boundary-1.tikz}} = \begin{tikzpicture}
	\begin{pgfonlayer}{nodelayer}
		\node [style=none] (1) at (-0.25, 0) {};
		\node [style=none] (2) at (1, 0) {};
		\node [style=permred] (3) at (-0.5, 0) {};
	\end{pgfonlayer}
	\begin{pgfonlayer}{edgelayer}
		\draw (1.center) to (2.center);
	\end{pgfonlayer}
\end{tikzpicture}
 \ ,\ \ \scalebox{.8}{\input{tikz/spin-boundary-3.tikz}} = \begin{tikzpicture}
	\begin{pgfonlayer}{nodelayer}
		\node [style=none] (1) at (-1, 0) {};
		\node [style=none] (2) at (0.25, 0) {};
		\node [style=permblue] (3) at (0.5, 0) {};
	\end{pgfonlayer}
	\begin{pgfonlayer}{edgelayer}
		\draw (1.center) to (2.center);
	\end{pgfonlayer}
\end{tikzpicture}
 \ \text{, and } \sum_{\scalebox{.5}{\begin{tikzpicture}
	\begin{pgfonlayer}{nodelayer}
		\node [style=permblue] (0) at (0, 0) {};
	\end{pgfonlayer}
\end{tikzpicture}
}\in S_2}\input{tikz/spin-boundary-7.tikz}=\frac{1}{20}\begin{tikzpicture}
	\begin{pgfonlayer}{nodelayer}
		\node [style=permred] (0) at (0.75, 0) {};
		\node [style=none] (1) at (0.5, 0) {};
		\node [style=none] (2) at (-0.75, 0) {};
	\end{pgfonlayer}
	\begin{pgfonlayer}{edgelayer}
		\draw (1.center) to (2.center);
	\end{pgfonlayer}
\end{tikzpicture}
 \, .
\]
\end{lemma}

Apply~\autoref{lem:single exp permutation} on a 1D circuit, the expectation of $\sum_{x\in\{0,1\}^n}q_C(x)^2$ is then exactly the sum over spin configurations on the the hexagonal lattice. For example,~\autoref{eq:1D} becomes the following. Note that each circle represents a distinct choices of elements from $S_2$.

\begin{equation}\label{eq:1D hexagon}
\Exp_{C\sim\cD}\left[\sum_{x\in\{0,1\}^n}q_C(x)^2\right] = 2^n\cdot\sum_{\scalebox{.5}{\begin{tikzpicture}
	\begin{pgfonlayer}{nodelayer}
		\node [style=permred] (0) at (0, 0) {};
	\end{pgfonlayer}
\end{tikzpicture}
},\scalebox{.5}{\begin{tikzpicture}
	\begin{pgfonlayer}{nodelayer}
		\node [style=permblue] (0) at (0, 0) {};
	\end{pgfonlayer}
\end{tikzpicture}
}\in S_2} \scalebox{.5}{\input{tikz/1D-hexagon.tikz}} \ .
\end{equation}

\paragraph{Step 2: Reducing to counting domain wall configurations on a triangular lattice}

The second idea in~\cite{Hunter19} is doing local summation on the blue vertices. Specifically, he showed that the local behavior of a blue vertex and its three red neighboring vertices can be fully described only by the spin of these three red vertices.

\begin{lemma}[{\cite[Equation~18]{Hunter19}}]\label{lem:exp plauette}
Let $U_g$ be a Haar-random $2$-qubit gate. We have the following.
\[
\sum_{\tau\in S_2}\input{tikz/plauette.tikz} = \begin{tikzpicture}
	\begin{pgfonlayer}{nodelayer}
		\node [style=plauette] (3) at (0, 0) {};
		\node [style=permred] (4) at (0.75, -1.25) {$\sigma_3$};
		\node [style=permred] (5) at (0.75, 1.25) {$\sigma_2$};
		\node [style=permred] (6) at (-1.5, 0) {$\sigma_1$};
	\end{pgfonlayer}
\end{tikzpicture}
 = \left\{\begin{array}{ll}
1     & ,\ \sigma_1=\sigma_2=\sigma_3 \\
0     & ,\ \sigma_1\neq\sigma_2=\sigma_3 \\
\frac{2}{5} & \textit{, else.}
\end{array}
\right.
\]
\end{lemma}

An immediate corollary of~\autoref{lem:exp plauette} is that now we can instead summing over the spin configurations over a triangular lattice. That is,~\autoref{eq:1D hexagon} becomes the following

\begin{equation}\label{eq:1D triangular spin}
\Exp_{C\sim\cD}\left[\sum_{x\in\{0,1\}^n}q_C(x)^2\right] =2^n\cdot\left(\frac{1}{20}\right)^{n/2}\cdot\sum_{\scalebox{.5}{\begin{tikzpicture}
	\begin{pgfonlayer}{nodelayer}
		\node [style=permred] (0) at (0, 0) {};
	\end{pgfonlayer}
\end{tikzpicture}
}\in S_2}\scalebox{.5}{\input{tikz/plauette-2.tikz}}
\end{equation}

The advantage of working on this triangular lattice is that the non-zero term on the right hand side of~\autoref{eq:1D triangular spin} corresponds to a \textit{domain wall} in the triangular lattice.

\begin{definition}[Domain wall]
Consider the right hand side of~\autoref{eq:1D triangular spin} and a configuration to all the red circles. The domain wall for this configuration is a collection of disjoint horizontal lines that separate the circles that are configured to $\mathbb{I}$ from the circles that are configured to $\mathbb{S}$.
\end{definition}

Note that the domain wall configuration is in 1-to-1 correspondence with the spin configurations. Let $\dw$ denote a domain wall, let $w(\dw)$ be the weight of the corresponding spin configuration. Thus,~\autoref{eq:1D triangular spin} becomes the following.

\begin{equation}\label{eq:1D triangular domain wall}
\Exp_{C\sim\cD}\left[\sum_{x\in\{0,1\}^n}q_C(x)^2\right] =\left(\frac{1}{5}\right)^{n/2}\cdot\sum_{\dw\in\scalebox{.1}{\input{tikz/plauette-2.tikz}}} w(\dw) \, .
\end{equation}

Note that a domain wall could contain two types of paths: (i) a path that goes from left boundary to the right boundary and (ii) a path that starts from and ends at both the right boundary. Furthermore, as the domain wall configuration is in 1-to-1 correspondence with subset of disjoint paths,~\autoref{eq:1D triangular domain wall} becomes the following.

\begin{equation}\label{eq:1D triangular domain wall path}
\Exp_{C\sim\cD}\left[\sum_{x\in\{0,1\}^n}q_C(x)^2\right]\leq\left(\frac{1}{5}\right)^{n/2}\cdot\left(1+\sum_{\substack{\text{disjoint paths \textsf{path} of type (i)}\\\in\scalebox{.1}{\input{tikz/plauette-2.tikz}}}} w(\textsf{path})\right)\cdot\left(1+\sum_{\substack{\text{disjoint paths \textsf{path} of type (ii)}\\\in\scalebox{.1}{\input{tikz/plauette-2.tikz}}}} w(\textsf{path})\right)\, .
\end{equation}

\paragraph{Step 3: Upper bound the sum of possible path configurations of type (i)}
For a set of disjoint paths \textsf{path} of type (i), it contains at most $\floor*{n/2}$ paths. Also, a path of type (i) contributes $\left(\frac{q}{q^2+1}\right)^d$ in the weight of a domain wall.

Now, for each $1\leq t\leq\floor*{n/2}$ let us first estimate the number of domain walls having at most $t$ paths of type (i). Specifically, for every $i\in[\floor*{n/2}]$, the number of paths starting from $i$ is at most $2^d$ because at each layer it either moves up or down.
Next, for each $1\leq t\leq\floor*{n/2}$, the number of possible $t$ paths is then at most $\binom{n/2}{t}\cdot2^{dt}$. This gives the following upper bound for the weight contributing from domain wall of type (i).
\begin{align*}
\left(\sum_{\substack{\text{disjoint paths \textsf{path} of type (i)}\\\in\scalebox{.1}{\input{tikz/plauette-2.tikz}}}} w(\textsf{path})\right)&\leq\sum_{t=1}^{\floor*{n/2}}\left(\frac{2}{5}\right)^{dt}\cdot\binom{\floor*{n/2}}{t}\cdot2^{dt}\\
&\leq\sum_{t=1}^{\floor*{n/2}}\left(\frac{4}{5}\right)^{dt}\cdot\binom{\floor*{n/2}}{t}\\
&\leq\left(1+\left(\frac{4}{5}\right)^d\right)^{\floor*{n/2}} \, .
\intertext{Consider $d\geq\frac{\log n}{\log(5/4)}$, the equation becomes}
&\leq\exp\left(\left(\frac{4}{5}\right)^{\frac{\log n}{\log(5/4)}}\frac{n}{2}\right)=O(1) \, .
\end{align*}

\paragraph{Step 4: Upper bound the sum of possible path configurations of type (ii)}
Let us consider the 1D circuit with infinite depth and denote the distribution as $\cD_\infty$.
Also, since the depth is infinity, the sum of the weight of domain wall with paths of type (i) is negligible. Namely, only paths of type (ii) contribute in the infinite depth circuit. Thus, we have the following upper bound.
\begin{align*}
\left(1+\sum_{\substack{\text{disjoint paths \textsf{path} of type (ii)}\\\in\scalebox{.1}{\input{tikz/plauette-2.tikz}}}} w(\textsf{path})\right)&\leq\left(1+\sum_{\substack{\text{disjoint paths \textsf{path} of type (ii)}\\\in\scalebox{.1}{\input{tikz/plauette-infty.tikz}}\cdots}} w(\textsf{path})\right)\\
&=5^{n/2}\cdot\Exp_{C\sim\cD_\infty}\left[\sum_{x\in\{0,1\}^n}q_C(x)^2\right] \, .
\intertext{Finally, it is a well known fact that the expectation of the sum of squares of marginal probabilities is $O(2^{-n})$ for infinite depth 1D circuit. So the above equation becomes}
&=O\left(\left(\frac{5}{4}\right)^{n/2}\right) \, .
\end{align*}

\paragraph{Wrap up}
To conclude the proof of~\autoref{lem:anti concentration of 1D}, let us plug in the calculations from step 3 and step 4 into~\autoref{eq:1D triangular domain wall path}, this gives us
\begin{align*}
\Exp_{C\sim\cD}\left[\sum_{x\in\{0,1\}^n}q_C(x)^2\right]&\leq\left(\frac{1}{5}\right)^{n/2}\cdot O(1)\cdot O\left(\left(\frac{5}{4}\right)^{n/2}\right)=O\left(\frac{1}{2^n}\right)
\end{align*}
as desired.

\paragraph{Acknowledgements.} We thank Scott Aaronson for helpful discussions.

\bibliography{mybib}
\bibliographystyle{alpha}

\appendix

\section{Useful properties of quantum circuits and tensor networks} \label{app:ommitted}

In this appendix we record some well-known and useful properties of quantum circuits and tensor networks. Readers could find more detailed discussion on tensor network in \cite{bridgeman2017hand} or \cite{gao2018efficient} for the applications to the simulation of quantum circuits.


\begin{lemma}\label{lem:sum of squares of marginal}
Let $C$ be a quantum circuit, $q_C$ be its pdf on input $\ket{0^n}$, and $U$ be the unitary matrix computed by $C$. For each $b\in\{0,1\}$, denote the marginal probability of the first output qubit being $b$ as $q_{C,b}$, we have
\[
\sum_{b\in\{0,1\}}q_{C,b}^2=\frac{1+\tr\left(Z\otimes I^{\otimes n-1}U^\dagger\ket{0^n}\bra{0^n}U\right)^2}{2}
\]
\end{lemma}
\begin{proof}[Proof of~\autoref{lem:sum of squares of marginal}]
Let $U$ be the $2^n\times2^n$ unitary matrix computed by the quantum circuit $C$, the square of the marginal probability (on input $\ket{0^n}$) can be calculated as follows. For each $b\in\{0,1\}$
\begin{align*}
q_{C,b}^2&=\left(\sum_{\substack{x\in\{0,1\}^n\\x_1=b}}\bra{0^n}U\ket{x}\bra{x}U^\dagger\ket{0^n}\right)^2 \, .
\intertext{By linearity and the fact that $\sum_{x\in\{0,1\}^n,\ x_1=n}\ket{x}\bra{x}=\ket{b}\bra{b}\otimes I^{\otimes n-1}$, the equation becomes}
&=\left(\bra{0^n}U\left(\ket{b}\bra{b}\otimes I^{\otimes n-1}\right)U^\dagger\ket{0^n}\right)^2=\tr\left(\left(\ket{b}\bra{b}\otimes I^{\otimes n-1}\right)U^\dagger\ket{0^n}\bra{0^n}U\right)^2 \, .
\end{align*}
Next, use the fact that $\ket{0}\bra{0}=\frac{I+Z}{2}$ and $\ket{1}\bra{1}=\frac{I-Z}{2}$ and denote $\Pi=U^\dagger\ket{0^n}\bra{0^n}U$, we have
\begin{align*}
\sum_{b\in\{0,1\}}q_{C,b}^2 &= \tr\left(\frac{I^{\otimes n}+Z\otimes I^{\otimes n-1}}{2}\Pi\right)^2+\tr\left(\frac{I^{\otimes n}-Z\otimes I^{\otimes n-1}}{2}\Pi\right)^2\\
&=\frac{\tr\left(I^{\otimes n}\Pi\right)^2+\tr\left(Z\otimes I^{\otimes n-1}\Pi\right)^2}{2}=\frac{1+\tr\left(Z\otimes I^{\otimes n-1}\Pi\right)^2}{2} \, .
\end{align*}
\end{proof}

\begin{lemma}[Some useful properties in tensor network for $2$-qubit quantum circuits]\label{lem:tensor net properties}
Let $A$ be the unitary matrix of a $2$-qubit gate, we have the following.
\begin{itemize}
\item (Trace) For any $2$-qubit unitary matrix $A$, we have
\begin{equation*}
    \input{tikz/trace.tikz} \ =\  \tr(A) \, .
\end{equation*}
\item (Matrix multiplication) For any $2$-qubit unitary matrices $A,B$, we have
\begin{equation*}
    \input{tikz/matrix-mult-1.tikz} \ =\  \input{tikz/matrix-mult-2.tikz} \, .
\end{equation*}
\item (Matrix vector multiplication) For any $2$-qubit state $\bra{\psi}$, we have
\begin{equation*}
    \input{tikz/matrix-vec-mult-1.tikz} \ =\  \bra{\psi}A \, .
\end{equation*}
\item (Pauli basis)
\begin{equation}\label{eq:pauli identity}
    \input{tikz/pauli-identity-1.tikz} = \frac{1}{2}\sum_{\sigma\in\{I,X,Y,Z\}}\input{tikz/pauli-identity-2.tikz} \, .
\end{equation}
\end{itemize}
\end{lemma}

\begin{lemma}[Pauli matrices in tensor network]\label{lem:pauli tensor net}
We have the following.
\begin{itemize}
\item (Orthogonality)
\[
\tr(\sigma\sigma') = \input{tikz/pauli-trace.tikz} = \left\{\begin{array}{ll}
2     & \text{, if }\sigma=\sigma' \\
0     & \text{, if }\sigma\neq\sigma' \, .
\end{array}
\right.
\]
\item (Evaluating on $\ket{0}$)
\[
\input{tikz/pauli-zero.tikz} = \left\{\begin{array}{ll}
1     & \text{, if }\sigma\in\{I,Z\} \\
0     & \text{, if }\sigma\in\{X,Y\} \, .
\end{array}
\right.
\]
\end{itemize}
\end{lemma}

\begin{lemma}[Disjoint light cone]\label{lem:disjoint light cone}
Let $C$ be a quantum circuit and $I=\{i_1,,i_2,\dots,i_m\}$ be output qubits having disjoint light cone. For each $x_I\in\{0,1\}^{|I|}$, denote the marginal probability of the output qubits in $I$ being measured to $x_I$, we have
\[
q_C(I,x_I) = \prod_{j=1}^mq_C\left(\{i_j\},x_{i_j}\right) \, .
\]
\end{lemma}
\begin{proof}[Proof of~\autoref{lem:disjoint light cone}]
Without loss of generality, let us assume $I=\{1,2,\dots,m\}$ and let $U$ denote the unitary matrix computed by $C$. The marginal probability of the output qubits in $I$ being measure to $x_I$ is
\begin{align*}
q_C(I,x_I) &= \left(\bra{0^n}U\ket{x_I}\otimes I^{\otimes n-m}\right)\left(\bra{x_I}\otimes I^{\otimes n-m}U^\dagger\ket{0^n}\right)\\
&= \scalebox{.8}{\input{tikz/disjoint.tikz}} \, .
\intertext{Note that the gate in $C$ that does not have a left-to-right path connecting to any of the output qubits in $x_I$ will be cancelled by its conjugate in the above tensor network (See below for an illustrating example). Thus, the tensor network breaks into $m$ disjoint tensor networks where each of them corresponds to the marginal probability of one of the output qubits in $I$. This gives us}
&=\prod_{j=1}^mq_C\left(\{i_j\},x_{i_j}\right)
\end{align*}
as desired.

As an illustrating example, consider the 1D circuit in~\autoref{fig:circuit skeleton} and $I=\{1,8\}$. For each $x_1,x_8\in\{0,1\}$, we have
\begin{align*}
q_C(I,x_I)&=\input{tikz/disjoint-example.tikz}\\
&=\input{tikz/disjoint-example-1.tikz}\\
&=q_C(\{1\},x_1)\cdot q_C(\{8\},x_8) \, .
\end{align*}

\end{proof}

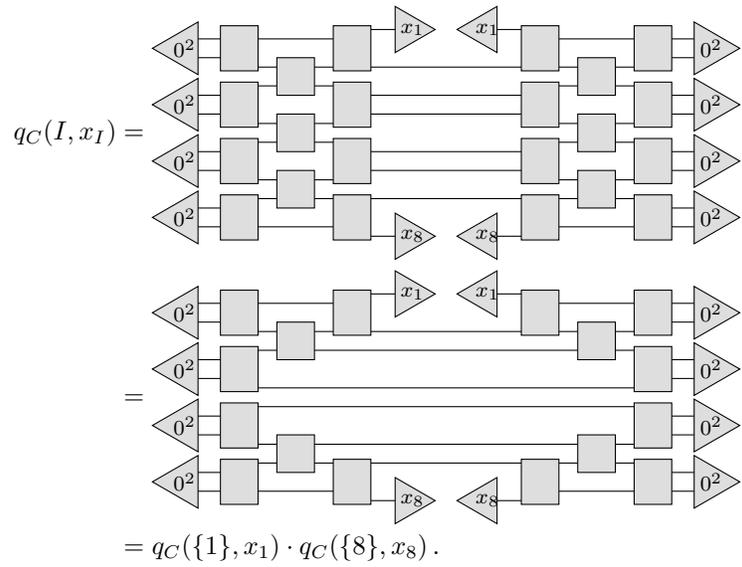
\end{document}